\documentclass[11pt]{article}
\usepackage{fullpage}
\usepackage{graphicx}
\usepackage{amsfonts} 
\usepackage{cite}
\usepackage[english]{babel}
\usepackage{amssymb,amsmath}
\usepackage{tikz}
\usetikzlibrary{shapes,arrows,shadows}
\usepackage{amssymb}
\usepackage{float} 
\usepackage{authblk} 

\newcommand{\ket}[1]{\mid #1 \rangle}

\newtheorem{theorem}{Theorem}
\newtheorem{lemma}[theorem]{Lemma}

\newtheorem{corollary}[theorem]{Corollary}
\newtheorem{remark}{Remark}

\newenvironment{proof}[1][Proof]{\begin{trivlist}
\item[\hskip \labelsep {\bfseries #1}]}{\end{trivlist}}
\newenvironment{definition}[1][Definition]{\begin{trivlist}
\item[\hskip \labelsep {\bfseries #1}]}{\end{trivlist}}

\newcommand{\qed}{\nobreak \ifvmode \relax \else
      \ifdim\lastskip<1.5em \hskip-\lastskip
      \hskip1.5em plus0em minus0.5em \fi \nobreak
      \vrule height0.75em width0.5em depth0.25em\fi}

\title{3-d quantum stabilizer codes with a power law energy barrier}			
\author{Kamil Michnicki }		
\date{\today}					
\affil{University of Washington, Department of Physics, Seattle, WA 98195-1560, USA}
\affil{kpm3@u.washington.edu}

\begin{document}
\maketitle						

\begin{abstract}
 We introduce a new primitive, called welding, for combining two stabilizer codes to produce a new stabilizer code. We apply welding to construct surface codes and then use the surface codes to construct solid codes, a variant of a 3-d toric code with rough and smooth boundaries.  Finally, we weld solid codes together to produce a $(O(L^3),1,O(L^{\frac{4}{3}}))$ stabilizer code with an energy barrier of $O(L^{\frac{2}{3}})$, which solves an open problem of whether a power law energy barrier is possible for local stabilizer code Hamiltonians in three-dimensions. The previous highest energy barrier is $O(\log L)$. Previous no-go results are avoided by breaking translation invariance.  
\end{abstract}

\tableofcontents

\section{Introduction}

The quantum analog of a hard disk drive is a self correcting quantum memory.  Whereas classical hard drives are in common use, it is currently unknown if a quantum hard drive can be built, even in principle.  A self-correcting quantum memory would protect quantum information from thermal noise by encoding a quantum state into the degenerate ground state of a Hamiltonian with a large energy barrier for local errors to produce logical errors. Natural thermalization with a cold reservoir would then keep the state of the system close to the ground state, provided the energy barrier were large enough. A theoretical understanding of self-correcting quantum memories would also give insight into whether topological phases of matter exist at non-zero temperature in three dimensions \cite{yoshida2011feasibility}.  The ferromagnetic hard disc drive, a good classical memory, has an energy penalty for neighboring magnetic domains that do not agree and thus to change the global polarization, locally, domains of magnetization must grow to encompass the entire material.  The total energy penalty is proportional to the perimeter of such domains, which accounts for the large energy barrier.  Self-correcting quantum memories have the further restriction that phase errors must also be protected against by a large energy barrier, not just errors in the computational basis. 

One could question the need for having a passive self-correcting quantum memory given the fact that doing active error correction on the toric code \cite{kitaev1997proceedings, kitaev2003fault}, one can construct an active memory whose storage time scales exponentially \cite{dennis2002topological} with the number of physical qubits used to encode a single protected qubit.  However, in order to do error correction on the toric code a nonlocal calculation must be done, i.e. one must find the minimum distance traversed by quasi-particles to annihilate. As we make the size of the code bigger, there is a point at which the classical processing becomes much longer than the necessary time step to perform error correction. Hence the exponential life-time, as a function of system size, has a limit. For this reason local error correction is necessary. One could do active error correction locally on the 2-d Ising model, a classical memory, by doing a majority vote within some finite neighborhood.  A similar statement holds for the 4-d toric code, which is known to be self-correcting \cite{dennis2002topological, alicki2008thermal}.  In both cases one uses the natural tension of the membrane operators to minimize the membrane locally.  Hence the existence of a high energy barrier is intimately linked with the ability to do active error correction locally.

The idea of the self-correcting quantum memory was first proposed in \cite{kitaev2003fault} and it was suggesting in \cite{dennis2002topological} that the 4-d toric code might be self correcting.  A 3-d subsystem code \cite{bacon2006operator} was hypothesized to be self correcting, though not confirmed. It was verified \cite{alicki2008thermal} that indeed the 4-d toric code has an exponential life-time when subjected to Markovian noise from a bath in the weak-coupling limit.  It was also shown \cite{castelnovo2007entanglement} that the 3-d toric code has two topological phase transitions, one at zero temperature corresponding to the string operators, and another at nonzero temperature, corresponding to the membrane operators. Thus the 3-d toric code can only act as a self-correcting classical memory.  It was proven \cite{bravyi2009no} that 2-d local self-correcting quantum memories are impossible for stabilizer codes. Another no-go theorem \cite{yoshida2011feasibility} was found for 3-d, local, translation-invariant stabilizer codes which have a bounded number of encoded qubits as a function of the system size.  Breaking the constraint on the number of qubits, a code with a $O(\log(L))$ energy barrier was found \cite{haah2011local, bravyi2011analytic,bravyi2011energy}. 

The current result is that a power law energy barrier can be reached for local stabilizer codes.

\begin{theorem}
\label{weldedCodeTheorem}
There exists a local stabilizer code Hamiltonian in 3-dimensions with qubits that fit inside of a box of side lengths $L$, that has an energy barrier of $O(L^{\frac{2}{3}})$.
\end{theorem}

Whether local quantum codes with power law energy barriers exist or not, has been an open problem before this result. In section 2 we discuss some of the background of stabilizer codes and energy barriers. In section 3, we develop a new technique, called welding, for creating new stabilizer codes from pre-existing stabilizer codes.  We apply welding, in section 4, to surface codes to illustrate the principles and introduce a technique for lower bounding energy barriers. Finally, in section 5 we apply welding to the 3-d toric code, with smooth and rough boundaries, to produce the desired code with a power law energy barrier.

\section{Preliminaries}
\subsection{Stabilizer Codes}
It is assumed that the reader is familiar with the stabilizer formalism for quantum error correcting codes. The following is a short review. For a more detailed review see \cite{nielsen2002quantum,gottesman1997stabilizer}.  

\begin{definition}[Pauli group:] The {\em Pauli group} $G_n$ on $n$ qubits is defined to be \begin{equation}
G_n=\{i^k \otimes_{i=1}^n P_i : k\in \{0,1,2,3\} \ \& \ P_j \in \{I,X,Y,Z\}\}
\end{equation}where $X,Y,Z$ are Pauli operators acting on a single qubit.
\end{definition}
 
\begin{definition}[Stabilizer group:] A {\em stabilizer group} $S$ on $n$ qubits is a subgroup of the Pauli group $G_n$ where $-I\notin S$.
\end{definition} 
Consequently $S$ is abelian, its coefficients are real and can be taken to be +1 without lack of generality.\footnote{This is because there exists a local unitary transformation in the Pauli group that transforms all the coefficients to +1.} 

\begin{definition}[Normalizer:] The {\em normalizer} of the stabilizer group is $N(S)=\{g\in G_n : gSg^{-1}=S\}$.
\end{definition}

\begin{definition}[Logical operators:] The set of {\em logical operators} is defined as $\{h:  h\in \eta S, \ \eta \in N(S), I\notin \eta S\}$. All logical operators in the same equivalence class $\eta S$ act identically on a state $\ket{\psi}$ when $\ket{\psi}$ has a +1 eigenvalue for each stabilizer in the stabilizer group S. The minimum weight of the set of logical operators is said to be the distance of the code.
\end{definition}

\begin{definition}[Stabilizer code Hamiltonian:]
Given a generating set $R$ for a stabilizer group $S$, we can encode the code space of $S$ in the ground state of the Hamiltonian $H=-\sum_{g\in R} g$. Such a Hamiltonian is called a {\em stabilizer code Hamiltonian}.
\end{definition}

This paper deals almost exclusively with CSS codes \cite{calderbank1996good,steane1996multiple}, which are codes with a generating set composed of generators that are either tensor products of Pauli-$X$ and identity operators or tensor products of Pauli-$Z$ and identity operators, called $X$-type and $Z$-type operators respectively. For such codes, $Z$ and $X$-type stabilizers commute when they overlap on an even number of qubits.

\begin{definition}[Standard CSS form:]
A generating set $R$ for a CSS stabilizer group  is in {\em standard CSS form} if each element of R is either a tensor product of only single-qubit Pauli-$X$ operators, e.g. $X\otimes I \otimes X...$ but not $X\otimes Z \otimes I...$, or a tensor product of single-qubit Pauli-$Z$ operators.
\end{definition}

\begin{definition}
For a generating set $R$ of a CSS code in standard CSS form, {\em $R^X$} is the set of all $X$-type stabilizers in $R$ and {\em $R^Z$} is the set of all $Z$-type stabilizers of $R$.
\end{definition}

\subsection{Energy Barriers}
  In order to define the energy barrier of a system we first have to define the interaction between the system and the environment, that is, we have to define an error model. Roughly speaking, the minimum energy penalty during any sequence of allowed errors that enacts a logical operator is called the energy barrier.  In this paper, we'll assume that the environment is periodically making local measurements of the stabilizers in the Hamiltonian and that furthermore, errors are local and the logical operators are nonlocal so that errors must be in the Pauli group.  We'll make the further restriction to single-qubit Pauli errors. A logical error is then generated by a Pauli sequence:

\begin{definition}[Pauli sequence:] A {\em Pauli sequence} $\{P_i\}$ is a sequence of single-qubit Pauli operators, $P_i$, applied to a set of qubits in order.
\end{definition}

If we apply a Pauli sequence resulting in the operator $P=P_n .... P_1$, where $P_i$ is a single qubit Pauli operator, the energy of the quantum state goes up, with respect to the ground state, for each term in the Hamiltonian that does not commute with $P$. This is because if $g\ket{\psi}=\ket{\psi}$ and $Pg=-gP$ then $gP\ket{\psi}=-Pg\ket{\psi}=-P\ket{\psi}$. We will refer to violated stabilizer generators as quasi-particles. In applying a Pauli sequence, we violate terms in the Hamiltonian and create quasi-particles and we would like to know how to apply a logical operator via a Pauli sequence in order to minimize the maximum number of quasi-particles throughout the sequence.

\begin{definition}[Energy barrier of a Pauli sequence:]  Given a generating set $R$ for the stabilizer group $S$ and a Hamiltonian $H=-\sum_{g\in R} g$, the {\em energy barrier of a Pauli sequence} is the maximum number of quasi-particles in existence throughout the Pauli sequence. 
\end{definition}

\begin{definition}[Energy barrier of a Pauli operator:] The {\em energy barrier of a Pauli operator} $p$ is the minimum-energy barrier over all Pauli sequences that produce $p$ with respect to a generating set $R$ for the Hamiltonian $H=-\sum_{g\in R} g$. 
\end{definition}

\begin{definition}[Energy barrier of a stabilizer code Hamiltonian:] The {\em energy barrier of a stabilizer code Hamiltonian} is the minimum over all energy barriers of logical operators in the code with respect to a generating set $R$ for the Hamiltonian $H=-\sum_{g\in R} g$.
\end{definition}

We made the restriction to Pauli sequences in our error model for if we allowed a local Pauli-group elements $e$ as an error in our error model then a Pauli sequence could generate $e$ with a constant-energy barrier. This would only change the energy barrier of a logical operator by an additive constant. 

For CSS codes we can make a further restriction to Pauli sequences of just Pauli-X or just Pauli-Z operators, assuming the Hamiltonian is of the form $H=H_X+H_Z$ where $H_X$ and $H_Z$ are sums of X-type and Z-type stabilizer respectively. This is because a Pauli-$Z$ error will only violate $X$-type stabilizers and a Pauli-$X$ error will only violate $Z$-type stabilizers. We'll use the convention in \cite{kitaev2003fault} where a violated $Z$-type stabilizer is called a $Z$-type quasi-particle and a violated $X$-type stabilizer is called an $X$-type quasi-particle. The restriction to Pauli-$Z$ sequences or Pauli-$X$ sequences will give the minimum-energy barrier for the code, for mixing can only increase the number of quasi-particles for a particular logical operator.

\section{Welded Codes}

How do we engineer a large energy barrier for logical operators in a local stabilizer code Hamiltonian?  Consider the $X$-type stabilizers of a CSS stabilizer code Hamiltonian. If $Z$ errors create single quasi-particles or pairs of quasi-particles then a $Z$-type logical operator will have a constant energy barrier because quasi-particles, once created, are free to move around without creating new quasi-particles, perhaps traveling through a nontrivial path and annihilating.  What we need are boundaries for which when $X$-type quasi-particles move past these boundaries, new $X$-type quasi-particles are created, that is, single qubit Pauli errors on these boundaries create three or more quasi-particles.  What's more is that these boundaries, where bifurcations happen, have to be unavoidable in the sense that quasi-particles must travel through a large number of them in creating a logical operator. If we want the code to be topological and have a high energy barrier then the $Z$-type quasi-particles must also have such boundaries where they create new quasi-particles in moving past these boundaries.  The trick is to simultaneously create such boundaries for both $X$-type and $Z$-type quasi-particles. One such trick is welding.

The rest of this section will be devoted to developing the theory of welding. We'll define two types of welding X-type and Z-type, give conditions for when welding can be understood simply in terms of carefully chosen generating sets and explain how to identify logical operators.

Like concatenation, (see \cite{gottesman1997stabilizer} for a discussion) welding can be used  to combine two codes, $S_1$ and $S_2$, to produce a third code, $S$. Whereas concatenation produces a code that acts on $N=n_1 n_2$ qubits when $S_1$ and $S_2$ act on $n_1$ and $n_2$ qubits respectively, welding produces a code that acts on $N < n_1+n_2$ qubits.  

Let's motivate the definition of welding with an example. Consider two codes $S_1=\langle XX,ZZ \rangle$ and $S_2=\langle XX,ZZ \rangle$ where $XX$ is shorthand for $X\otimes X$.  Suppose we identify qubit two of $S_1$ with qubit one of $S_2$ so that we have three qubits total.  Properly, after the identification $S_1=\langle XXI,ZZI \rangle$ and $S_2=\langle IXX,IZZ \rangle$.  $\langle S_1\cup S_2 \rangle$ is not a stabilizer group since $(XXI)(IZZ)=-(IZZ)(XXI)$.  However if we adopt the $X$-type operators from $S_1$ and $S_2$ we can ask the question, which $Z$-type operators commute with the $X$-type operators?  In this simple case $ZZZ$ is the only one.  We define $S_1 \boxplus S_2 := \langle XXI,IXX,ZZZ \rangle$ where $\boxplus$ denotes a $Z$-type weld. Similarly we could have adopted $Z$-type operators and updated $X$-type operators in which case we would define $S_1 \boxtimes S_2 := \langle XXX,ZZI,IZZ\rangle$. This is called an $X$-type weld, as $X$-type stabilizers appear to be welded together.  Let's focus on the $Z$-type weld.  $ZZZ$ can be viewed as the operators $ZZI$ and $IZZ$ welded together on the second qubit and hence it is called a welded operator. Let $\theta_i(p)$ be the restriction of the Pauli operator $p$, after the identification, to the qubits that $S_i$ acts nontrivially on.  In this case $\theta_1(ZZZ)=ZZI$ and $\theta_2(ZZZ)=IZZ$. That $\theta_i(h)\in S_i$ for a $Z$-type operator $h$ is a general property for $Z$ type welds.  Also, if we define $W(h)$ to be the restriction onto qubits shared between the two codes we see that $h=\theta_1(h)\theta_2(h)W(h)$ where in this case $h=ZZZ$ and $W(h)=IZI$. This is also a general property of welding. Finally notice that if we remove $ZZZ$ from $S_1 \boxplus S_2$ and promote it to be a logical operator, then we have the repetition code, encoding one qubit with three qubits that corrects a single phase error. Some of the definitions introduced in this paragraph are important to this and further sections and so they are summarized in table \ref{keydefinitions}.

\begin{table}[H]
\begin{tabular}{|l|l|}
\hline
$\boxplus$ & A $Z$-type weld, where $Z$-type stabilizers are welded together. \\
\hline 
$\boxtimes$ & An $X$-type weld, where $X$-type stabilizers are welded together. \\
\hline
$\theta_i(p)$ & Restriction to qubits where stabilizers of $S_i$ act nontrivially. \\
\hline
$W(p)$ & Restriction to qubits shared between $S_1$ and $S_2$ after the identification. \\
\hline
\end{tabular}
\label{keydefinitions}
\caption{Important definitions pertaining to welding.}
\end{table}

The concept of welding codes can be applied to codes other than CSS codes and can also be applied to welding codes to themselves, i.e. by identifying qubits within a code and updating the stabilizer group. For the purpose of this paper we will only consider welding between CSS codes that each encode zero qubits. We consider zero encoded qubits because including logical operators into the stabilizer group provides for a more compact description of welding. We  regain logical operators after welding by choosing an operator from the generating set of the welded code to promote to a logical operator. The general procedure for welding CSS codes is as follows:

\begin{definition}[Z-type weld:] For two CSS codes $S_1$ and $S_2$ in standard CSS form the process of welding involves:
\end{definition}
\begin{itemize}
\item[1] Identifying distinct pairs of qubits between $S_1$ and $S_2$ where for each pair we contract that pair to a single qubit.
\item[2] Adopt all $X$-type operators of $S_1$ and $S_2$ into a new stabilizer group $S$.
\item[3] Include, into the stabilizer group S, any $Z$-type operator that commutes with all the $X$-type operators from $S_1$ and $S_2$.
\item[4] $S=S_1\boxplus S_2$ is generated by all such elements from steps 2 and 3.
\end{itemize}

Step 3 can be simplified by noting that if a $Z$-type operator g commutes with $X$-type operators of both $S_1$ and $S_2$, then $\theta_1(g) =g_1\in S_1$ and $\theta_2(g)=g_2\in S_2$ and so $g=\theta_1(g)\theta_2(g)W(g)=g_1 g_2 W(g_1)$ where $W(g_1)=W(g_2)$. The $Z$-type stabilizers of the welded code are all such combinations, thus the name "welding". An equivalent version of step number 3 is:

\begin{itemize}
\item[3'] For $Z$-type operators $g_1\in S_1^Z$ and $g_2 \in S_2^Z$ where $W(g_1)=W(g_2)$, include all $Z$-type operators $g_1 g_2 W(g_1)$ into the welded code $S$.
\end{itemize}

An $X$-type weld can be defined similarly to a $Z$-type weld where for an $X$-type weld we adopt $Z$-type operators and adopt all $X$-type operators that commute with the $Z$-type operators. The theorems in this section will be for Z-type welds as the symmetric theorems are automatically true for X-type welds.

\begin{theorem}
There are zero qubits encoded in the welded code if the codes $S_1$ and $S_2$ also have zero qubits encoded.  
\end{theorem}
\begin{proof}
The welded code constitutes a stabilizer group of the CSS kind.  If there were any encoded qubits after welding, then there would be an encoded $Z$ operator for each encoded qubit. By definition, these encoded $Z$ operators commute with all of the $X$-type stabilizers, which contradicts our assumption that we had included all such $Z$-type operators. $\Box$
\end{proof}

For a $Z$-type weld, the fact that we're including all $Z$-type operators that commute with the $X$-type operators seems like a lot of work and in fact we can reduce the work by picking generating sets of $S_1$ and $S_2$ wisely and then welding operators in the generating sets to produce a generating set for the welded code.  The generating sets have to be in a special form though and we'll need some definitions before we can state this form concisely.

\begin{definition}[Well Matched:] When generating sets $R_1$ and $R_2$ of $S_1$ and $S_2$, respectively, are in standard CSS form and match up on the weld, this is called {\em well matched}. More specifically, $R_1$ and $R_2$ are well matched if for every operator $h\in R_1^Z$ such that $W(h)\neq I$ there exists $g\in R_2^Z$ such that $W(h)=W(g)$. Similarly for every $h\in R_2^Z$ such that $W(h)\neq I$ there exists $g\in R_1^Z$ such that $W(h)=W(g)$.  
\end{definition}

\begin{definition}[Linearly independence on the weld:] $R_k$ is {\em linearly independent on the weld} if $R_k$ is in standard CSS form and for any independent subset of $Z$-type stabilizers $A=\{h\}\subset R_k^Z$, if $W(h)\neq I$ $\forall h \in A$ then $W(\prod_{h \in A}h)\neq I$. 
\end{definition}

Linear independence on the weld basically means that you can't multiply elements that act non-trivially on the weld in order to get an element that acts trivially on the weld, otherwise the welded code would not necessarily encode zero qubits.

\begin{theorem} 
\label{weldedGenSet}
Let $S_1$ and $S_2$ be two CSS codes that are welded together to produce $S$ via a $Z$-type weld.  If two generating sets $R_1$ and $R_2$ of $S_1$ and $S_2$, respectively, are in standard CSS form, well matched and linearly independent on the weld, then we can find a generating set $R$ of $S$ such that 1.) if $h\in R_i^X$ then $h\in R$ and 2.) if $h_1\in R_1^Z$ and $h_2 \in R_2^Z$ such that $W(h_1)=W(h_2)$ then $h_1 h_2 W(h_1) \in R$.
\end{theorem}

\begin{proof}
Trivially, all $X$-type operators are generated. If g is a $Z$-type operator of the welded code $S$ then we can find an operator $h$, generated by $R$ such that $\theta_1(hg)=I$ because of the well matched condition. $hg\in S_2$ and $W(hg)=I$ and by linear independence on the weld, $hg$ must be generated by elements of $R_2$ that are identity on the weld, which are also in $R$. $\Box$
\end{proof}

\begin{remark} 
\label{anticommutingOperator}
Suppose $R$ is a generating set which is a product of doing a $Z$ weld between standard CSS form, well matched, linearly independent on the weld generating sets $R_1$ and $R_2$. If $h_i \in R_i^Z$ and $t$ is an $X$-type operator that anticommutes with $h_i$ and commutes with all other operators in $R_i$ then t will only anticommute with a welded operators of the form $h\in R^Z$ where $\theta_i(h)=h_i$.
\end{remark}

Theorem \ref{weldedGenSet} says that we can weld together generating sets of stabilizer groups to produce a generating set for the welded code, which is important if we want to either efficiently perform the welding operation or visualize welding. Remark \ref{anticommutingOperator} says that if an operator anticommutes with a stabilizer generator before welding, it will anticommute with it after it is welded.  This is important because after we have a generating set for the welded code we will choose an operator from the generating set to promote to a logical operator and this remark tells us how to quickly identifying anticommuting logical operators.

\section{Surface Codes and Welding}

\subsection{Surface Codes}

To illustrate the process of welding, we generate the surface code \cite{dennis2002topological} by welding only groups of the form $G=\langle XX,ZZ\rangle$ together. Going back to the example of the previous section, we have two groups $S_1=\langle XX,ZZ \rangle$ and $S_2=\langle XX,ZZ\rangle$ and we identify qubit two of $S_1$ with qubit one of $S_2$ and we do a $Z$-type weld. After the embedding into 3 qubits, we have $S_1=\langle XXI,ZZI \rangle$ and $S_2=\langle IXX,IZZ\rangle$. We get the group $S_1 \boxplus S_2=\langle XXI,IXX,ZZZ \rangle$ where $ZZZ=(ZZI)(IZZ)W(ZZI)$.  This is shown diagrammatically in figure \ref{repcode}.

\begin{figure}[H]
\begin{tikzpicture}[scale=1]

\draw (1,1.2)node {$\alpha$};
\draw (2,1.2)node {$\alpha$};
\draw (5,1.2)node {$\alpha$};

\path(0,0) node[draw,shape=circle,scale=0.4,fill] (v0) {};
\path(1,0) node[draw,pink,shape=circle,scale=1] (v1) {};
\path(0,1) node[draw,shape=circle,scale=1] (v2) {};
\path(1,1) node[draw,shape=circle,scale=0.4,fill] (v3) {};

\draw[dashed,pink] (v0)--(v1)--(v3);
\draw[blue,thick] (v0)--(v2)--(v3);

\path(3,1) node[draw,shape=circle,scale=1] (v3) {};
\path(2,1) node[draw,shape=circle,scale=0.4,fill] (v2) {};
\path(3,0) node[draw,shape=circle,scale=0.4,fill] (v1) {};
\path(2,0) node[draw,pink,shape=circle,scale=1] (v0) {};

\draw[dashed,pink] (v2)--(v0)--(v1);
\draw[blue,thick] (v1)--(v3)--(v2);

\draw (1.5,.5)node {$\boxplus$};
\draw (3.5,.5)node {$=$};

\path(4,0) node[draw,shape=circle,scale=0.4,fill] (v0) {};
\path(5,0) node[draw,pink,shape=circle,scale=1] (v1) {};
\path(4,1) node[draw,shape=circle,scale=1] (v2) {};
\path(5,1) node[draw,shape=circle,scale=0.4,fill] (v3) {};
\path(6,1) node[draw,shape=circle,scale=1] (v4) {};
\path(6,0) node[draw,shape=circle,scale=0.4,fill] (v5) {};

\draw[dashed,pink] (v0)--(v1)--(v3) (v1)--(v5);
\draw[blue,thick] (v0)--(v2)--(v3)--(v4)--(v5);

\end{tikzpicture}
        
\caption{Welding two, two qubit repetition codes together.  Black dots are qubits and circles are stabilizers. Each stabilizer vertex acts on adjacent qubit vertices with a Pauli operator corresponding to the type of edge that connects the two. A solid edge indicates a Pauli-$X$ operator being applied and a dashed edge indicates a Pauli-$Z$ operator being applied. $\alpha$ indicates which qubits are being identified.}
\label{repcode}
\end{figure}
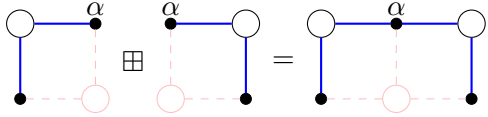

Continuing on diagrammatically, we weld two such codes together, via an $X$-type weld, as in figure \ref{smallestSurfaceCode}.

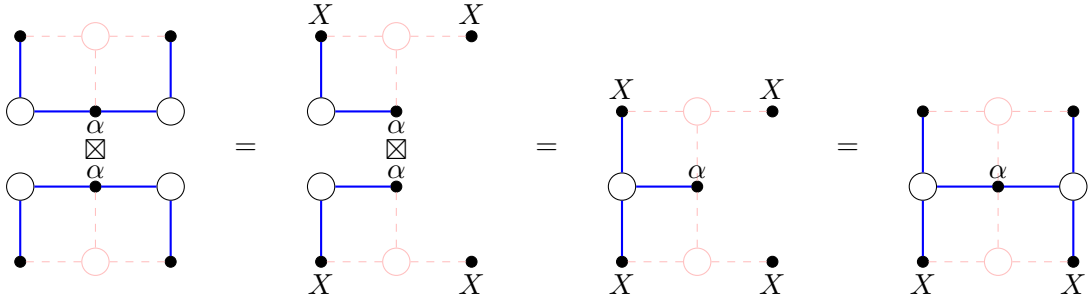
\begin{figure}[H]
\begin{tikzpicture}[scale=1]

\path(0,0) node[draw,shape=circle,scale=0.4,fill] (v00) {};
\path(1,0) node[draw,pink,shape=circle,scale=1] (v10) {};
\path(2,0) node[draw,shape=circle,scale=0.4,fill] (v20) {};
\path(0,1) node[draw,shape=circle,scale=1] (v01) {};
\path(1,1) node[draw,shape=circle,scale=0.4,fill] (v11) {};
\path(2,1) node[draw,shape=circle,scale=1] (v21) {};

\draw[dashed,pink] (v00)--(v10)--(v20) (v10)--(v11);
\draw[blue,thick] (v00)--(v01)--(v11)--(v21)--(v20);

\path(0,2) node[draw,shape=circle,scale=1] (v00) {};
\path(1,2) node[draw,shape=circle,scale=0.4,fill] (v10) {};
\path(2,2) node[draw,shape=circle,scale=1] (v20) {};
\path(0,3) node[draw,shape=circle,scale=0.4,fill] (v01) {};
\path(1,3) node[draw,pink,shape=circle,scale=1] (v11) {};
\path(2,3) node[draw,shape=circle,scale=0.4,fill] (v21) {};

\draw[dashed,pink] (v01)--(v11)--(v21) (v11)--(v10);
\draw[blue,thick] (v01)--(v00)--(v10)--(v20)--(v21);

\draw (1,1.2)node {$\alpha$};
\draw (1,1.8)node {$\alpha$};
\draw (1,1.5)node {$\boxtimes$};
\draw(3,1.5)node {$=$};


\path(4,0) node[draw,shape=circle,scale=0.4,fill] (v00) {};
\path(5,0) node[draw,pink,shape=circle,scale=1] (v10) {};
\path(6,0) node[draw,shape=circle,scale=0.4,fill] (v20) {};
\path(4,1) node[draw,shape=circle,scale=1] (v01) {};
\path(5,1) node[draw,shape=circle,scale=0.4,fill] (v11) {};

\draw[dashed,pink] (v00)--(v10)--(v20) (v10)--(v11);
\draw[blue,thick] (v00)--(v01)--(v11);

\path(4,2) node[draw,shape=circle,scale=1] (v00) {};
\path(5,2) node[draw,shape=circle,scale=0.4,fill] (v10) {};
\path(4,3) node[draw,shape=circle,scale=0.4,fill] (v01) {};
\path(5,3) node[draw,pink,shape=circle,scale=1] (v11) {};
\path(6,3) node[draw,shape=circle,scale=0.4,fill] (v21) {};

\draw[dashed,pink] (v01)--(v11)--(v21) (v11)--(v10);
\draw[blue,thick] (v01)--(v00)--(v10);

\draw (5,1.2)node {$\alpha$};
\draw (5,1.8)node {$\alpha$};
\draw (5,1.5)node {$\boxtimes$};
\draw (4,-.3)node {$X$};
\draw (6,-.3)node {$X$};
\draw (4,3.3)node {$X$};
\draw (6,3.3)node {$X$};

\draw(7,1.5)node {$=$};


\path(8,0) node[draw,shape=circle,scale=0.4,fill] (v00) {};
\path(9,0) node[draw,pink,shape=circle,scale=1] (v10) {};
\path(10,0) node[draw,shape=circle,scale=0.4,fill] (v20) {};
\path(8,1) node[draw,shape=circle,scale=1] (v01) {};
\path(9,1) node[draw,shape=circle,scale=0.4,fill] (v11) {};
\path(8,2) node[draw,shape=circle,scale=0.4,fill] (v02) {};
\path(9,2) node[draw,pink,shape=circle,scale=1] (v12) {};
\path(10,2) node[draw,shape=circle,scale=0.4,fill] (v22) {};

\draw[dashed,pink] (v00)--(v10)--(v20) (v10)--(v11) (v02)--(v12)--(v22) (v12)--(v11);
\draw[blue,thick] (v00)--(v01)--(v02) (v01)--(v11);

\draw (9,1.2)node {$\alpha$};
\draw (8,-.3)node {$X$};
\draw (10,-.3)node {$X$};
\draw (8,2.3)node {$X$};
\draw (10,2.3)node {$X$};
\draw(11,1.5)node {$=$};


\path(12,0) node[draw,shape=circle,scale=0.4,fill] (v00) {};
\path(13,0) node[draw,pink,shape=circle,scale=1] (v10) {};
\path(14,0) node[draw,shape=circle,scale=0.4,fill] (v20) {};
\path(12,1) node[draw,shape=circle,scale=1] (v01) {};
\path(13,1) node[draw,shape=circle,scale=0.4,fill] (v11) {};
\path(12,2) node[draw,shape=circle,scale=0.4,fill] (v02) {};
\path(13,2) node[draw,pink,shape=circle,scale=1] (v12) {};
\path(14,2) node[draw,shape=circle,scale=0.4,fill] (v22) {};
\path(14,1) node[draw,shape=circle,scale=1] (v21) {};

\draw[dashed,pink] (v00)--(v10)--(v20) (v10)--(v11) (v02)--(v12)--(v22) (v12)--(v11);
\draw[blue,thick] (v00)--(v01)--(v02) (v01)--(v11) (v20)--(v21)--(v22) (v21)--(v11);

\draw (13,1.2)node {$\alpha$};
\draw (12,-.3)node {$X$};
\draw (14,-.3)node {$X$};

\end{tikzpicture}

        
\caption{$\alpha$ identifies the two qubits to be welded. Requiring linear independence on the weld invokes a horizontal string of $X$ operators to be in the stabilizer group after the weld. After the second step we have specified two horizontal $X$-type string operators by placing an $X$ on the qubit vertex of the stabilizer.  The reason for doing this is so that it is in the same form as an arbitrary surface code which has much longer string like operators that can't conveniently be written as a stabilizer vertex.}
\label{smallestSurfaceCode}
\end{figure}

 After welding we have a 5 qubit surface code. But the string like operator, the horizontal $XX$ operator, is not yet three qubits long and so would not have protection from $Z$ errors if it were promoted to a logical operator.  In the next diagram we continue welding these tiny surface codes together, extending the length of the string like operators. We find that welding two 5 qubit surface codes creates a 7 qubit surface code.

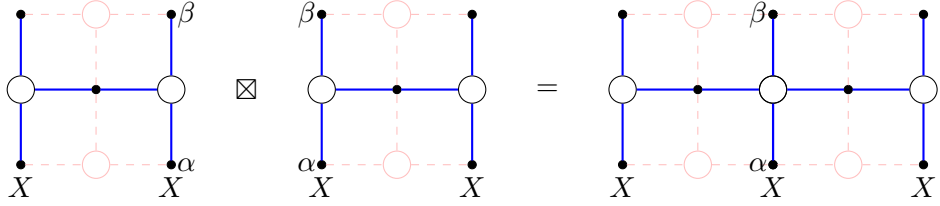
\begin{figure}[H]

\begin{tikzpicture}[scale=1] 
 \path(0,0) node[draw,shape=circle,scale=0.3,fill] (v00) {};
    \path(1,0) node[draw,shape=circle,pink,scale=1] (v10) {};
    \path(2,0) node[draw,shape=circle,scale=0.3,fill] (v20){}; 
    \path(0,1) node[draw,shape=circle,scale=1] (v01) {}; 
    \path(1,1) node[draw,shape=circle,scale=0.3,fill] (v11) {}; 
    \path(2,1) node[draw,shape=circle,scale=1] (v21) {}; 
    \path(0,2) node[draw,shape=circle,scale=0.3,fill] (v02) {}; 
    \path(1,2) node[draw,shape=circle,pink,scale=1] (v12) {}; 
    \path(2,2) node[draw,shape=circle,scale=0.3,fill] (v22) {}; 
    \draw[dashed,pink,scale=1] (v00)--(v10)--(v20) (v10)--(v11)--(v12) (v02)--(v12)--(v22); 
    \draw[blue,thick,scale=1] (v00)--(v01)--(v02) (v01)--(v11)--(v21) (v20)--(v21)--(v22); 
    
        \path(4,0) node[draw,shape=circle,scale=0.3,fill] (v00) {};
    \path(5,0) node[draw,shape=circle,pink,scale=1] (v10) {};
    \path(6,0) node[draw,shape=circle,scale=0.3,fill] (v20){}; 
    \path(4,1) node[draw,shape=circle,scale=1] (v01) {}; 
    \path(5,1) node[draw,shape=circle,scale=0.3,fill] (v11) {}; 
    \path(6,1) node[draw,shape=circle,scale=1] (v21) {}; 
    \path(4,2) node[draw,shape=circle,scale=0.3,fill] (v02) {}; 
    \path(5,2) node[draw,shape=circle,pink,scale=1] (v12) {}; 
    \path(6,2) node[draw,shape=circle,scale=0.3,fill] (v22) {}; 
    \draw[dashed,pink,scale=1] (v00)--(v10)--(v20) (v10)--(v11)--(v12) (v02)--(v12)--(v22); 
    \draw[blue,thick,scale=1] (v00)--(v01)--(v02) (v01)--(v11)--(v21) (v20)--(v21)--(v22); 
    
        \path(8,0) node[draw,shape=circle,scale=0.3,fill] (v00) {};
    \path(9,0) node[draw,shape=circle,pink,scale=1] (v10) {};
    \path(10,0) node[draw,shape=circle,scale=0.3,fill] (v20){}; 
    \path(8,1) node[draw,shape=circle,scale=1] (v01) {}; 
    \path(9,1) node[draw,shape=circle,scale=0.3,fill] (v11) {}; 
    \path(10,1) node[draw,shape=circle,scale=1] (v21) {}; 
    \path(8,2) node[draw,shape=circle,scale=0.3,fill] (v02) {}; 
    \path(9,2) node[draw,shape=circle,pink,scale=1] (v12) {}; 
    \path(10,2) node[draw,shape=circle,scale=0.3,fill] (v22) {}; 
    \draw[dashed,pink,scale=1] (v00)--(v10)--(v20) (v10)--(v11)--(v12) (v02)--(v12)--(v22); 
    \draw[blue,thick,scale=1] (v00)--(v01)--(v02) (v01)--(v11)--(v21) (v20)--(v21)--(v22); 
    
        \path(10,0) node[draw,shape=circle,scale=0.3,fill] (v00) {};
    \path(11,0) node[draw,shape=circle,pink,scale=1] (v10) {};
    \path(12,0) node[draw,shape=circle,scale=0.3,fill] (v20){}; 
    \path(10,1) node[draw,shape=circle,scale=1] (v01) {}; 
    \path(11,1) node[draw,shape=circle,scale=0.3,fill] (v11) {}; 
    \path(12,1) node[draw,shape=circle,scale=1] (v21) {}; 
    \path(10,2) node[draw,shape=circle,scale=0.3,fill] (v02) {}; 
    \path(11,2) node[draw,shape=circle,pink,scale=1] (v12) {}; 
    \path(12,2) node[draw,shape=circle,scale=0.3,fill] (v22) {}; 
    \draw[dashed,pink,scale=1] (v00)--(v10)--(v20) (v10)--(v11)--(v12) (v02)--(v12)--(v22); 
    \draw[blue,thick,scale=1] (v00)--(v01)--(v02) (v01)--(v11)--(v21) (v20)--(v21)--(v22);

\draw (3,1)node {$\boxtimes$};
\draw (7,1)node {$=$};
\draw (2.2,0)node {$\alpha$};
\draw (2.2,2)node {$\beta$};
\draw (3.8,0)node {$\alpha$};
\draw (3.8,2)node {$\beta$};
\draw (9.8,0)node {$\alpha$};
\draw (9.8,2)node {$\beta$};
\draw (0,-.3)node {$X$};
\draw (2,-.3)node {$X$};
\draw (4,-.3)node {$X$};
\draw (6,-.3)node {$X$};
\draw (8,-.3)node {$X$};
\draw (10,-.3)node {$X$};
\draw (12,-.3)node {$X$};

        \end{tikzpicture}

\caption{Welding two surface codes together on their smooth edge.  The strings of $X$s become welded together and promoted to a logical operator.}
\label{halfcorrecting}
\end{figure}

Using remark \ref{anticommutingOperator} we see that a vertical $ZZ$ string anti-commutes with the horizontal string of three $X$s in figure \ref{halfcorrecting}. Switching these in the stabilizer group generating set, represented by the graph, we can do a $Z$-type weld between two strips of surface code to arrive at a 13 qubit surface code that has string logical operators that protect against an arbitrary error on a single qubit as in figure \ref{fullcorrecting}.

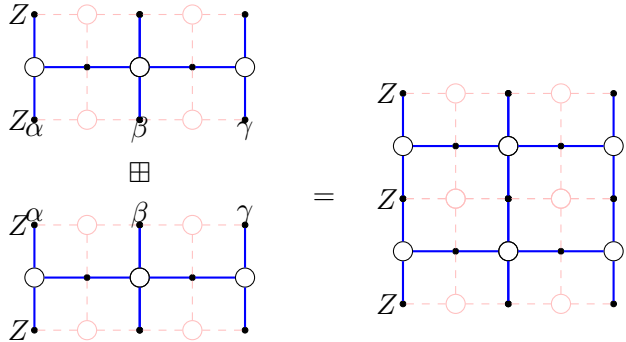
\begin{figure}[H]
\begin{tikzpicture}[scale=0.7] 
 \path(0,0) node[draw,shape=circle,scale=0.21,fill] (v00) {};
    \path(1,0) node[draw,shape=circle,pink,scale=0.7] (v10) {};
    \path(2,0) node[draw,shape=circle,scale=0.21,fill] (v20){}; 
    \path(0,1) node[draw,shape=circle,scale=0.7] (v01) {}; 
    \path(1,1) node[draw,shape=circle,scale=0.21,fill] (v11) {}; 
    \path(2,1) node[draw,shape=circle,scale=0.7] (v21) {}; 
    \path(0,2) node[draw,shape=circle,scale=0.21,fill] (v02) {}; 
    \path(1,2) node[draw,shape=circle,pink,scale=0.7] (v12) {}; 
    \path(2,2) node[draw,shape=circle,scale=0.21,fill] (v22) {}; 
    \draw[dashed,pink,scale=0.7] (v00)--(v10)--(v20) (v10)--(v11)--(v12) (v02)--(v12)--(v22); 
    \draw[blue,thick,scale=0.7] (v00)--(v01)--(v02) (v01)--(v11)--(v21) (v20)--(v21)--(v22); 
    
        \path(2,0) node[draw,shape=circle,scale=0.21,fill] (v00) {};
    \path(3,0) node[draw,shape=circle,pink,scale=0.7] (v10) {};
    \path(4,0) node[draw,shape=circle,scale=0.21,fill] (v20){}; 
    \path(2,1) node[draw,shape=circle,scale=0.7] (v01) {}; 
    \path(3,1) node[draw,shape=circle,scale=0.21,fill] (v11) {}; 
    \path(4,1) node[draw,shape=circle,scale=0.7] (v21) {}; 
    \path(2,2) node[draw,shape=circle,scale=0.21,fill] (v02) {}; 
    \path(3,2) node[draw,shape=circle,pink,scale=0.7] (v12) {}; 
    \path(4,2) node[draw,shape=circle,scale=0.21,fill] (v22) {}; 
    \draw[dashed,pink,scale=0.7] (v00)--(v10)--(v20) (v10)--(v11)--(v12) (v02)--(v12)--(v22); 
    \draw[blue,thick,scale=0.7] (v00)--(v01)--(v02) (v01)--(v11)--(v21) (v20)--(v21)--(v22); 
    
        \path(0,4) node[draw,shape=circle,scale=0.21,fill] (v00) {};
    \path(1,4) node[draw,shape=circle,pink,scale=0.7] (v10) {};
    \path(2,4) node[draw,shape=circle,scale=0.21,fill] (v20){}; 
    \path(0,5) node[draw,shape=circle,scale=0.7] (v01) {}; 
    \path(1,5) node[draw,shape=circle,scale=0.21,fill] (v11) {}; 
    \path(2,5) node[draw,shape=circle,scale=0.7] (v21) {}; 
    \path(0,6) node[draw,shape=circle,scale=0.21,fill] (v02) {}; 
    \path(1,6) node[draw,shape=circle,pink,scale=0.7] (v12) {}; 
    \path(2,6) node[draw,shape=circle,scale=0.21,fill] (v22) {}; 
    \draw[dashed,pink,scale=0.7] (v00)--(v10)--(v20) (v10)--(v11)--(v12) (v02)--(v12)--(v22); 
    \draw[blue,thick,scale=0.7] (v00)--(v01)--(v02) (v01)--(v11)--(v21) (v20)--(v21)--(v22); 
    
        \path(2,4) node[draw,shape=circle,scale=0.21,fill] (v00) {};
    \path(3,4) node[draw,shape=circle,pink,scale=0.7] (v10) {};
    \path(4,4) node[draw,shape=circle,scale=0.21,fill] (v20){}; 
    \path(2,5) node[draw,shape=circle,scale=0.7] (v01) {}; 
    \path(3,5) node[draw,shape=circle,scale=0.21,fill] (v11) {}; 
    \path(4,5) node[draw,shape=circle,scale=0.7] (v21) {}; 
    \path(2,6) node[draw,shape=circle,scale=0.21,fill] (v02) {}; 
    \path(3,6) node[draw,shape=circle,pink,scale=0.7] (v12) {}; 
    \path(4,6) node[draw,shape=circle,scale=0.21,fill] (v22) {}; 
    \draw[dashed,pink,scale=0.7] (v00)--(v10)--(v20) (v10)--(v11)--(v12) (v02)--(v12)--(v22); 
    \draw[blue,thick,scale=0.7] (v00)--(v01)--(v02) (v01)--(v11)--(v21) (v20)--(v21)--(v22); 
    
        \path(7,0.5) node[draw,shape=circle,scale=0.21,fill] (v00) {};
    \path(8,0.5) node[draw,shape=circle,pink,scale=0.7] (v10) {};
    \path(9,0.5) node[draw,shape=circle,scale=0.21,fill] (v20){}; 
    \path(7,1.5) node[draw,shape=circle,scale=0.7] (v01) {}; 
    \path(8,1.5) node[draw,shape=circle,scale=0.21,fill] (v11) {}; 
    \path(9,1.5) node[draw,shape=circle,scale=0.7] (v21) {}; 
    \path(7,2.5) node[draw,shape=circle,scale=0.21,fill] (v02) {}; 
    \path(8,2.5) node[draw,shape=circle,pink,scale=0.7] (v12) {}; 
    \path(9,2.5) node[draw,shape=circle,scale=0.21,fill] (v22) {}; 
    \draw[dashed,pink,scale=0.7] (v00)--(v10)--(v20) (v10)--(v11)--(v12) (v02)--(v12)--(v22); 
    \draw[blue,thick,scale=0.7] (v00)--(v01)--(v02) (v01)--(v11)--(v21) (v20)--(v21)--(v22); 
    
        \path(7,2.5) node[draw,shape=circle,scale=0.21,fill] (v00) {};
    \path(8,2.5) node[draw,shape=circle,pink,scale=0.7] (v10) {};
    \path(9,2.5) node[draw,shape=circle,scale=0.21,fill] (v20){}; 
    \path(7,3.5) node[draw,shape=circle,scale=0.7] (v01) {}; 
    \path(8,3.5) node[draw,shape=circle,scale=0.21,fill] (v11) {}; 
    \path(9,3.5) node[draw,shape=circle,scale=0.7] (v21) {}; 
    \path(7,4.5) node[draw,shape=circle,scale=0.21,fill] (v02) {}; 
    \path(8,4.5) node[draw,shape=circle,pink,scale=0.7] (v12) {}; 
    \path(9,4.5) node[draw,shape=circle,scale=0.21,fill] (v22) {}; 
    \draw[dashed,pink,scale=0.7] (v00)--(v10)--(v20) (v10)--(v11)--(v12) (v02)--(v12)--(v22); 
    \draw[blue,thick,scale=0.7] (v00)--(v01)--(v02) (v01)--(v11)--(v21) (v20)--(v21)--(v22); 
    
        \path(9,0.5) node[draw,shape=circle,scale=0.21,fill] (v00) {};
    \path(10,0.5) node[draw,shape=circle,pink,scale=0.7] (v10) {};
    \path(11,0.5) node[draw,shape=circle,scale=0.21,fill] (v20){}; 
    \path(9,1.5) node[draw,shape=circle,scale=0.7] (v01) {}; 
    \path(10,1.5) node[draw,shape=circle,scale=0.21,fill] (v11) {}; 
    \path(11,1.5) node[draw,shape=circle,scale=0.7] (v21) {}; 
    \path(9,2.5) node[draw,shape=circle,scale=0.21,fill] (v02) {}; 
    \path(10,2.5) node[draw,shape=circle,pink,scale=0.7] (v12) {}; 
    \path(11,2.5) node[draw,shape=circle,scale=0.21,fill] (v22) {}; 
    \draw[dashed,pink,scale=0.7] (v00)--(v10)--(v20) (v10)--(v11)--(v12) (v02)--(v12)--(v22); 
    \draw[blue,thick,scale=0.7] (v00)--(v01)--(v02) (v01)--(v11)--(v21) (v20)--(v21)--(v22); 
    
        \path(9,2.5) node[draw,shape=circle,scale=0.21,fill] (v00) {};
    \path(10,2.5) node[draw,shape=circle,pink,scale=0.7] (v10) {};
    \path(11,2.5) node[draw,shape=circle,scale=0.21,fill] (v20){}; 
    \path(9,3.5) node[draw,shape=circle,scale=0.7] (v01) {}; 
    \path(10,3.5) node[draw,shape=circle,scale=0.21,fill] (v11) {}; 
    \path(11,3.5) node[draw,shape=circle,scale=0.7] (v21) {}; 
    \path(9,4.5) node[draw,shape=circle,scale=0.21,fill] (v02) {}; 
    \path(10,4.5) node[draw,shape=circle,pink,scale=0.7] (v12) {}; 
    \path(11,4.5) node[draw,shape=circle,scale=0.21,fill] (v22) {}; 
    \draw[dashed,pink,scale=0.7] (v00)--(v10)--(v20) (v10)--(v11)--(v12) (v02)--(v12)--(v22); 
    \draw[blue,thick,scale=0.7] (v00)--(v01)--(v02) (v01)--(v11)--(v21) (v20)--(v21)--(v22);

\draw (2,3)node {$\boxplus$};
\draw (5.5,2.5)node {$=$};
\draw (0,2.2)node {$\alpha$};
\draw (2,2.2)node {$\beta$};
\draw (4,2.2)node {$\gamma$};

\draw (0,3.8)node {$\alpha$};
\draw (2,3.8)node {$\beta$};
\draw (4,3.8)node {$\gamma$};
\draw (-.3,0)node {$Z$};
\draw (-.3,2)node {$Z$};
\draw (-.3,4)node {$Z$};
\draw (-.3,6)node {$Z$};

\draw (6.7,.5)node {$Z$};
\draw (6.7,2.5)node {$Z$};
\draw (6.7,4.5)node {$Z$};

        \end{tikzpicture}
\label{fullcorrecting}
\caption{Welding two surface codes together on their rough edge also welds the logical $Z$ operators from each surface code together, creating a longer string of $Z$s.}
\end{figure}

By first welding a thin strip of surface code and then welding those strips together, we can generate a surface code of an arbitrary size.  

From these examples it should be clear that one finds the shape of the welded logical operators by first including them into the stabilizer group, welding them and then promoting them to logical operators again.  In doing so we are always welding codes with zero encoded qubits to generate a code with zero encoded qubits. For each independent stabilizer generator that we promote to a logical operator, we get an additional encoded qubit. If we promote the string operators of the surface code to logical operators then we have a code that encodes a single logical qubit. 

It should be noted that the top and bottom edges of the surface code in figure \ref{fullcorrecting} are called rough edges and that the left and right edges are called smooth edges. Rough edges are where a single $X$-type quasi-particles can be created from a single $Z$ error and smooth edges are where a single $Z$-type quasi-particle can be created from a single $X$ error.

\subsection{Welded Surface Codes}


In this section we specialize to welding boundaries of surface codes. Specifically, smooth edges to smooth edges with $X$-type welds and rough edges to rough edges with $Z$-type welds. We show examples of how to create a high energy barrier and then show how regions where quasi-particles move around freely are used to put a lower bound on the energy barrier.

In the previous section we showed how to weld two surface codes to produce a larger surface code, however, the larger surface code still has string logical operators. Let's consider what happens if instead we welded three surface codes together along a rough edge with a $Z$-type weld as in figure \ref{roughWeldOn3Surfaces}.   More precisely, we first weld the first two surface codes together and then weld the third surface code to the product. It will be useful to describe the rough and smooth boundaries of a given surface code after it has been welded.

\begin{definition}
A {\em rough weld} describes the set of qubits along the rough edge of a surface code after it has been welded on that rough edge via a $Z$-type weld. A {\em smooth weld} is defined similarly where it describes the set of identified qubits of an $X$-type weld on a smooth edge of a surface code.
\end{definition}

\begin{figure}[H]
\input {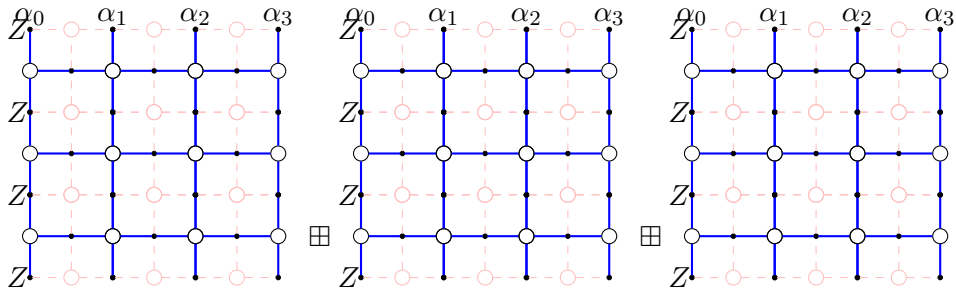}
\caption{Three surface codes welded together by their rough edge via a $Z$-type weld. $\alpha_i$ label the identified qubits that will be welded together.}
\label{roughWeldOn3Surfaces}
\end{figure}

Before the weld we include, into the stabilizer group, the logical operator composed of a string of $Z$ operators extending between opposite rough edges.  We are then welding three codes with zero logical qubits together to produce a code with zero logical qubits.  One of the generators of the welded code will be the three strings of $Z$ operators joined on the line of $\alpha$s, which we'll call $\bar{Z}$.  If we promote $\bar{Z}$ to a logical operator, we will then have a code that encodes a single qubit and $\bar{Z}$ has an energy barrier of 2. 

Notice that on a particular surface, if there is an odd number of $Z$ errors on the rough boundaries, i.e. rough welds and rough edges, of that surface, then there will be at least one quasi-particle in that surface. Also notice that $\bar{Z}$ has only one Pauli-$Z$ operator on each rough boundary of each surface and that stabilizers of our welded code act with even numbers of Pauli $Z$s on the rough boundaries so that any representation of $\bar{Z}$ will have an odd number of Pauli $Z$s on their rough boundaries.  For the example in figure \ref{roughWeldOn3Surfaces}, at least one Pauli Z has to be applied to the rough weld, hence there is no way to produce $\bar{Z}$, via a Pauli-$Z$ sequence, without creating at least two quasi-particles. This bound can be saturated by considering an $X$-type quasi-particle being created at a rough edge and traveling through the rough weld creating an additional $X$-type quasi-particle. Both quasi-particles then annihilate at their respective rough edges. 

What about $Z$-type quasi-particles?  Applying an $X$ error on one of the smooth edges creates a single $Z$-type quasi-particle but everywhere else $Z$-type quasi-particles get created in pairs.  This is an important observation because it means that $Z$-type quasi-particles can pass through $Z$-type welds on rough edges without creating new quasi-particles.

Similarly we can do $X$-type welds on smooth edges of surface codes as in figure \ref{smoothWeldOn3Surfaces} in order to increase the energy barrier for the $\bar{Z}$ operator of the encoded qubit.  Actually, we could first weld three surface codes via a rough edge as in figure \ref{roughWeldOn3Surfaces}, and then weld three such codes together via an $X$-type weld on a common smooth boundary in order to create a code that has an energy barrier of 2 for both $\bar{X}$ and $\bar{Z}$.  

This simple example illustrates the principle behind how to generate energy barriers using welding. We make boundaries where errors create three or more quasi-particles and then those quasi-particles must travel through some region to annihilate at a different boundary.

\begin{figure}[H]
\input {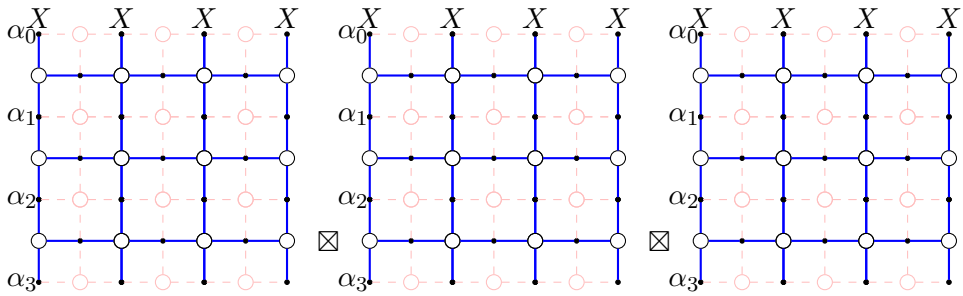}
\caption{Three surface codes welded together on a smooth edge with an $X$-type weld.}
\label{smoothWeldOn3Surfaces}
\end{figure}

Some key facts to remember about $Z$-type welds on rough edges and $X$-type welds on smooth edges are as follows:

\begin{itemize}
\item An $X$-type quasi-particle splits in going past a rough weld.
\item $Z$-type quasi-particles pass through rough welds without splitting.
\item A $Z$-type quasi-particle splits in going past a smooth weld.
\item $X$-type quasi-particles pass through smooth welds without splitting.
\end{itemize}

These four rules allow us to find regions where quasi-particles can move around without creating new quasi-particles or annihilating. 

\begin{definition}[Flat regions:] Regions of qubits where $X$ and $Z$-type quasi-particles can move around without creating new quasi-particles will be called {\em flat-$X$ regions and flat-$Z$ regions} respectively.
\end{definition}

If we only did $Z$-type welds between rough edges of surface codes then a flat-$X$ region of our code, call it $A$, will have rough boundaries, i.e. a set of rough welds, for which a $Z$ error on a rough weld will create a single $X$-type quasi-particle in region $A$. If there are an odd number of $Z$ errors on all the rough boundaries of A then there will be at least one $X$-type quasi-particle in $A$ for otherwise an even number of quasi-particles could annihilate. In general we can split our code into several flat-$X$ regions, call them $A_i$, that will be connected via rough welds where a Pauli-$Z$ error on a rough weld creates a single excitation in each adjacent region $A_i$.  If a particular representation of a logical $Z$ operator has an odd number of single qubit $Z$ operators on its rough welds, then throughout a Pauli sequence generating $\bar{Z}$, for each moment when $A_i$ has an odd number of $Z$ errors on its rough welds there will be at least one quasi-particle in $A_i$.  We can encode this information via a graph where each flat-$X$ region is a vertex that connects to a "rough" vertex via an edge.  A similar graph can be formed for flat-$Z$ regions and smooth boundaries.

\begin{figure}[h!]
   \includegraphics[width=70mm]{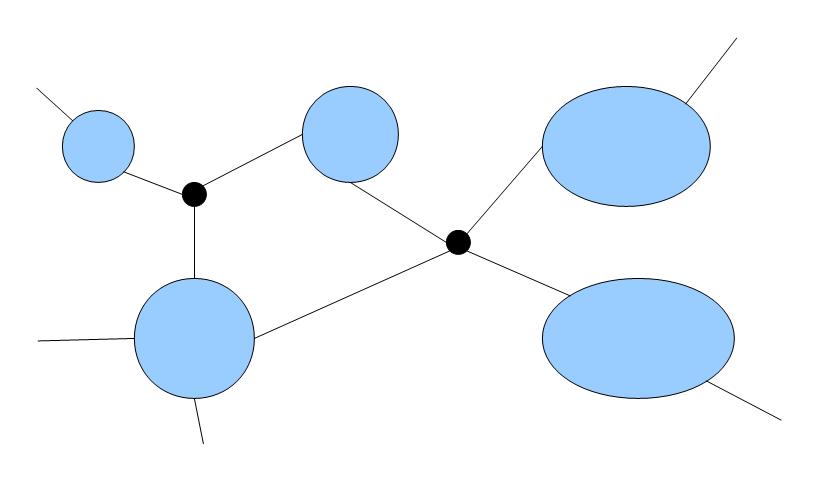}
\caption{An abstract diagram of flat-$X$ regions, represented as large nodes, connected by rough boundaries, represented as dots in black. A $Z$ error on a rough boundary creates an $X$-type quasi-particle in each connected flat-$X$ region.}
\end{figure}

\begin{lemma}[Parity lower bound]
\label{parityLowerBound}
\label{XZregions}
For a code that is the result of welding surface codes together via $X$-type welds on smooth boundaries and $Z$-type welds on rough boundaries, let $\{A_i\}$ be the set of flat $X$ regions of the code and $p_i$ be the parity of the number of $Z$ errors on the rough boundaries of $A_i$.  For a Pauli-$Z$ walk, the total number of quasi-particles $|F_k|\geq \sum_i p_i$ where $|F_k|$ is the total number of $X$-type quasi-particles at step $k$ of the Pauli-$Z$ walk. A similar statement holds for $Z$-type quasi-particles and flat-$Z$ regions.
\end{lemma}

  Notice that the sum is not taken modulo 2, that is, if three regions have parity 1 then the sum will be at least 3. The parity lower bound lemma tells us how to put a lower bound on the energy barrier of welded codes and in fact, this bound is saturated if pairs of quasi-particles in the same regions annihilate as soon as they are created in those regions.

\begin{proof}[Proof of the parity lower bound:]
An $X$-type quasi-particle can travel through $X$-type welds on smooth edges without creating new quasi-particles and so if an even number of $X$-type quasi-particles exists within a given flat $X$ region, they can annihilate whereas if an odd number of errors occur on the rough boundary of a flat X region, an odd number of $X$-type quasi-particles must exist within that flat $X$ region with a minimum of one $Z$-type quasi-particle existing. Summing over all such regions we find that the total number of quasi-particles satisfies $|F_k|\geq \sum_i p_i$ where $p_i$ is the parity of $Z$ errors on the boundaries the $i$th flat-$X$ region. The proof of the parity lower bound for flat-$Z$ regions follows similarly. $\Box$
\end{proof}

\section{Solid Codes and Welding}
\subsection{Solid Codes}

In this section we introduce the solid code, a homology code \cite{bombin2009quantum}, which is the 3-d analog of surface codes, i.e. a 3-d toric code \cite{castelnovo2007entanglement} with rough and smooth surfaces, and prove a lower bound on the energy barrier of the solid code by analyzing the parity of quasi-particles in its flat regions.

Define the graph of a solid code by a cubic lattice.  Next we remove all horizontal edges at the top and bottom boundary of the cube.  Qubits live on edges. The the vertices of the cubic lattice are $X$-type stabilizers where for each vertex the stabilizer is the product of Pauli-X operators for each edge connected to that vertex and identity elsewhere. We do not include the vertices at the top and bottom rough boundaries as stabilizers.  For each plaquette of this graph include a $Z$-type stabilizer that acts on the qubits of that face.  On the roughened boundary we have half plaquettes, like for the surface code, where the $Z$-type stabilizer is the product of three Pauli-$Z$ operators. A small version of a solid code is shown in figure \ref{smallSolid}.  

\begin{figure}[H]
\includegraphics[width=140mm]{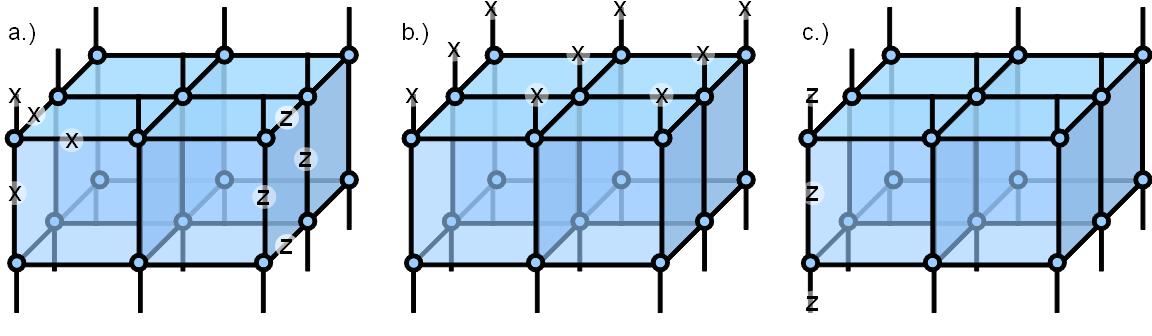}
\caption{The graph of a small solid code with qubits on the edges. a.)  An $X$-star operator and a $Z$-plaquette operator are shown. b.) The $\bar{X}$-membrane operator is shown. c.)The $\bar{Z}$-string operator is shown. }
\label{smallSolid}
\end{figure}

\begin{theorem}
There is one encoded qubit for the solid code with logical operators that are a membrane of Pauli-$X$ operators and a string of Pauli-$Z$ operators.
\end{theorem}

\begin{proof}
 We'll prove this by constructing the solid code by welding surface codes together and then promoting a single stabilizer to a logical operator to arrive at the solid code. To start, consider a tall thin surface code as in figure \ref{tallThinSurfaceCode} with a horizontal-$XX$ operator in the generating set. This strip can be as tall as we wish. Label the left and right smooth edge with $i$ and $j$ and label the surface code with $(i,j)$.  Now do an $X$-type weld on these smooth edges such that the smooth edges and surface codes form a graph $G=(V,E)$ of a square grid where $V=\{i\}$ and $E=\{(i,j)\}$.  The tall thin surface code with a horizontal-$XX$ string operator encodes zero qubits because the tall thin surface codes are well matched and linearly independent on the weld so that welding them together in a square grid also produces a code with zero qubits by theorem \ref{weldedGenSet}.  The horizontal-$XX$ operators get welded together to form a membrane of Pauli-$X$ operators, which we call $\bar{X}$. Promoting $\bar{X}$ to a logical operator gives a code with a single encoded qubit. 

The logical operator that anitcommutes with $\bar{X}$, that is $\bar{Z}$, is the string of $Z$s going between the top and bottom of each surface code. We can see this as a consequence of remark \ref{anticommutingOperator}. 

The given code does not have horizontal plaquette operators but notice that the horizontal $Z$-plaquette operators are in the stabilizer group because we can multiply four half plaquettes at a rough boundary to get a plaquette operator. We can then multiply this horizontal plaquette by four adjacent vertical plaquettes to move this horizontal plaquette up or down.  $\Box$
\end{proof}

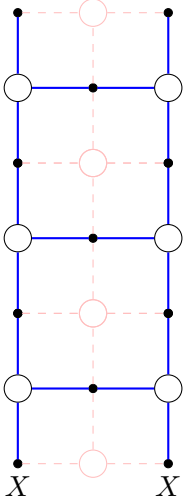
\begin{figure}[H]
\begin{tikzpicture}[scale=1] 
 \path(0,0) node[draw,shape=circle,scale=0.3,fill] (v00) {};
    \path(1,0) node[draw,shape=circle,pink,scale=1] (v10) {};
    \path(2,0) node[draw,shape=circle,scale=0.3,fill] (v20){}; 
    \path(0,1) node[draw,shape=circle,scale=1] (v01) {}; 
    \path(1,1) node[draw,shape=circle,scale=0.3,fill] (v11) {}; 
    \path(2,1) node[draw,shape=circle,scale=1] (v21) {}; 
    \path(0,2) node[draw,shape=circle,scale=0.3,fill] (v02) {}; 
    \path(1,2) node[draw,shape=circle,pink,scale=1] (v12) {}; 
    \path(2,2) node[draw,shape=circle,scale=0.3,fill] (v22) {}; 
    \draw[dashed,pink,scale=1] (v00)--(v10)--(v20) (v10)--(v11)--(v12) (v02)--(v12)--(v22); 
    \draw[blue,thick,scale=1] (v00)--(v01)--(v02) (v01)--(v11)--(v21) (v20)--(v21)--(v22); 
    
        \path(0,2) node[draw,shape=circle,scale=0.3,fill] (v00) {};
    \path(1,2) node[draw,shape=circle,pink,scale=1] (v10) {};
    \path(2,2) node[draw,shape=circle,scale=0.3,fill] (v20){}; 
    \path(0,3) node[draw,shape=circle,scale=1] (v01) {}; 
    \path(1,3) node[draw,shape=circle,scale=0.3,fill] (v11) {}; 
    \path(2,3) node[draw,shape=circle,scale=1] (v21) {}; 
    \path(0,4) node[draw,shape=circle,scale=0.3,fill] (v02) {}; 
    \path(1,4) node[draw,shape=circle,pink,scale=1] (v12) {}; 
    \path(2,4) node[draw,shape=circle,scale=0.3,fill] (v22) {}; 
    \draw[dashed,pink,scale=1] (v00)--(v10)--(v20) (v10)--(v11)--(v12) (v02)--(v12)--(v22); 
    \draw[blue,thick,scale=1] (v00)--(v01)--(v02) (v01)--(v11)--(v21) (v20)--(v21)--(v22); 
    
        \path(0,4) node[draw,shape=circle,scale=0.3,fill] (v00) {};
    \path(1,4) node[draw,shape=circle,pink,scale=1] (v10) {};
    \path(2,4) node[draw,shape=circle,scale=0.3,fill] (v20){}; 
    \path(0,5) node[draw,shape=circle,scale=1] (v01) {}; 
    \path(1,5) node[draw,shape=circle,scale=0.3,fill] (v11) {}; 
    \path(2,5) node[draw,shape=circle,scale=1] (v21) {}; 
    \path(0,6) node[draw,shape=circle,scale=0.3,fill] (v02) {}; 
    \path(1,6) node[draw,shape=circle,pink,scale=1] (v12) {}; 
    \path(2,6) node[draw,shape=circle,scale=0.3,fill] (v22) {}; 
    \draw[dashed,pink,scale=1] (v00)--(v10)--(v20) (v10)--(v11)--(v12) (v02)--(v12)--(v22); 
    \draw[blue,thick,scale=1] (v00)--(v01)--(v02) (v01)--(v11)--(v21) (v20)--(v21)--(v22);

\draw (0,-.3)node {$X$};
\draw (2,-.3)node {$X$};

        \end{tikzpicture}
\caption{An example of a tall thin surface code.}
\label{tallThinSurfaceCode}
\end{figure}

Next we will find the energy barrier of $\bar{X}$ and $\bar{Z}$. First notice that we can produce $\bar{Z}$ via a Pauli sequence with a constant energy barrier because $X$-type quasi-particles can move without creating new quasi-particles when away from the rough boundaries, that is, $\bar{Z}$ is a string operator. 

In order to put a lower bound on the energy barrier of the $\bar{X}$-membrane operator we'll remove the horizontal-$Z$-plaquette operators from the generating set as in the construction via welding. Each tall thin surface $(i,j)\in E$ then becomes a flat-$Z$ region where $Z$-type quasi-particles can move around without creating new $Z$-type quasi-particles.  We can use the parity lower bound lemma to lower bound the energy barrier. The flat-$Z$ regions and smooth boundaries are shown graphically in figure \ref{SolidZtypeRegions}. The lower bound is given by the 2-d Ising model which means that:

\begin{lemma}
\label{solidCodeEnergyBarrier}
For a solid code that is $O(d)$ qubits wide, the energy barrier for $\bar{X}$  is $O(d)$.  
\end{lemma}

\begin{proof}
We use the parity lower bound on the flat-X regions of the solid code with horizontal plaquettes removed. The flat-$X$ regions are in the shape of a 2-d ising model. This gives an energy barrier of $O(d)$. This bound can be saturated if pairs of $Z$ type quasi-particles are annihilated as soon as they are created in each flat-$Z$ region. Adding horizontal plaquettes, as in the 3-d toric code\footnote{This same proof technique, via flat regions, applies equally to the 3-d toric code \cite{castelnovo2007entanglement} to provide a lower bound for the energy barrier of $O(d)$ for a 3-d toric code with $O(d^3)$ qubits.}., could only increase the energy and even then, keeping the membrane operator completely horizontal, the bound $O(d)$ can be saturated. $\Box$
\end{proof}

\begin{figure}[H]
\includegraphics[width=50mm]{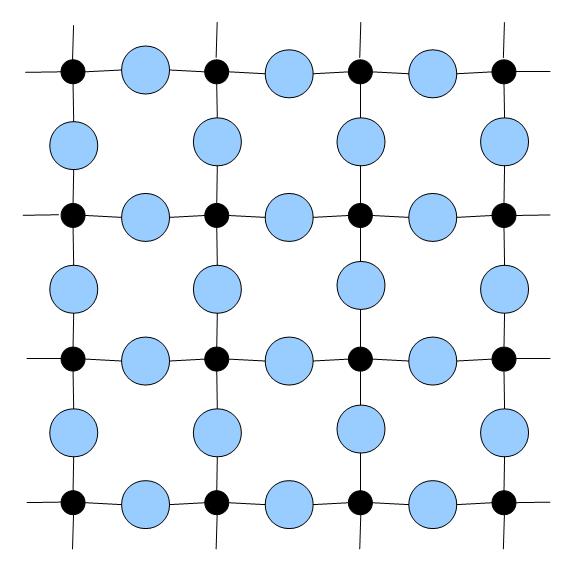}
\caption{The blue regions are flat-$Z$ regions and black vertices are smooth edges. The energy barrier is bounded below by the energy barrier of the Ising model corresponding to the same graph as this figure with the reinterpretation that the black dots are qubits and the blue circles are $ZZ$ operators in the hamiltionian $H=\sum_{(i,j)\in E} -Z_i Z_j$.}
\label{SolidZtypeRegions}
\end{figure}

\subsection{Welded Solid Code}

In this section the basic objects that we'll be welding together are the solid codes with generating sets that do not include the horizontal-$Z$ plaquettes, as they could only increase the energy barrier. As a first example we'll weld three solid codes together by their rough boundaries, discuss the generalization of this and then analyze the energy barrier of the resulting code.

As a simple example consider welding the three solid codes in figure \ref{threeSmallWeldedSolids} by their rough boundaries with a $Z$-type weld.  After welding, the $\bar{Z}$ operator becomes three strings of $Z$s welded together at the rough boundary. $\bar{Z}$ will have a minimum energy barrier of 2 for the same reason that three surface codes welded on a rough edge have an energy barrier of 2. An $X$-type quasi-particle created at a rough boundary of one of the solids must annihilate on a different rough boundary and so must pass through the rough weld and create an additional $X$-type quasi-particle. 

\begin{figure}[H]
\includegraphics[width=140mm]{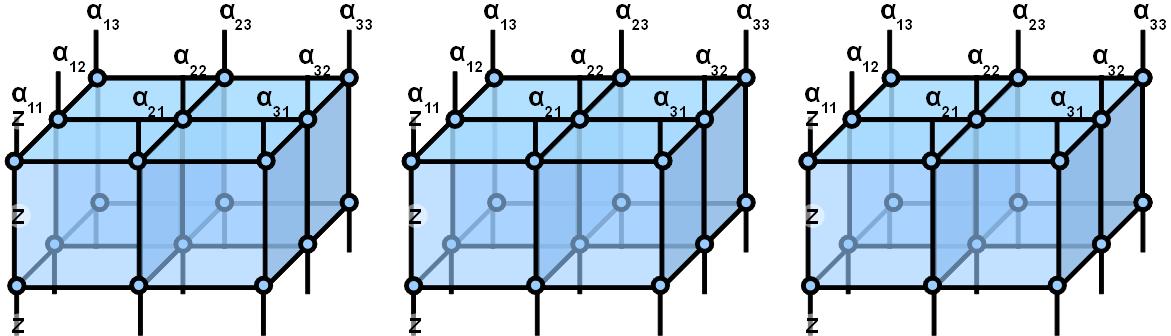}
\caption{Three small solid codes with a $Z$-type weld between their upper rough boundary. The $\alpha_{ij}$ symbols label identical qubits.}
\label{threeSmallWeldedSolids}
\end{figure}

Notice that the $\bar{X}$ operator does not change after welding as can be seen from remark \ref{anticommutingOperator}. It is still a membrane in one of the solids. The graph of flat-$Z$ regions does not change after the welding because $Z$-type quasi-particles can propagate through $Z$-type welds on rough boundaries without creating new $Z$-type quasi-particles. Hence the energy barrier for the $\bar{X}$ operator stays the same before and after the weld.

In analogy to welding surface codes together on smooth edges to make a solid code with an increased energy barrier for the $\bar{X}$ operator, we can weld solid codes together on rough boundaries to increase the energy barrier of the $\bar{Z}$ operator.  To start we label the rough boundaries of a solid code by $i$ and $j$ and then label the solid by its pair of boundaries $(i,j)$.  We can then weld the solid $(i,j)$ with the solid $(n,k)$ by identifying qubits of the boundary j with qubits of the boundary n and then doing a $Z$-type weld.  We are left with solids $(i,j)$ and $(j,k)$ welded on the boundary j. In this way we can weld solid codes to form any graph $G=(V,E)$. Figure \ref{2dWeldedSolidCode} shows a 2-d square graph of solids welded together. A 3-d cubic graph of solids gives a better energy barrier but is harder to show graphically.

\begin{figure}[H]
\includegraphics[width=70mm]{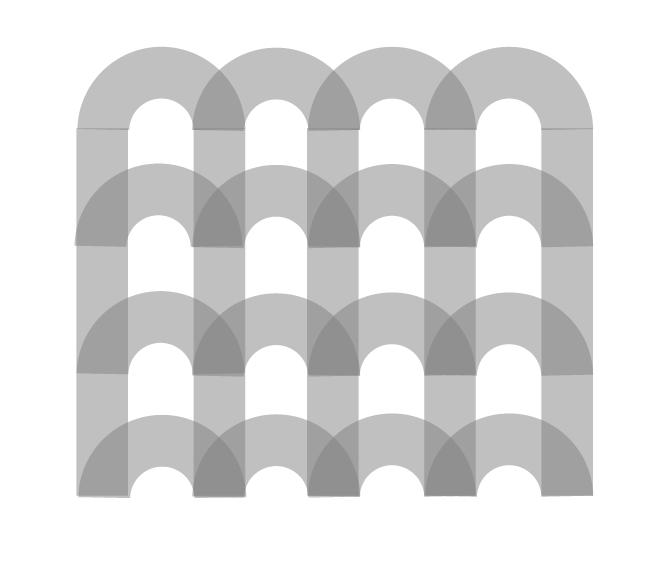}
\includegraphics[width=70mm]{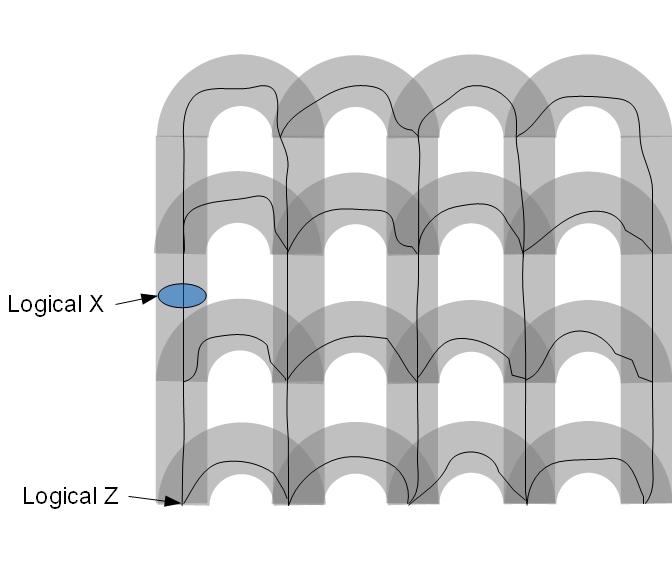}
\caption{Welded solid code with solids welded together in a 2d square lattice. Notice that the object as a whole is three dimensional.}
\label{2dWeldedSolidCode}
\end{figure}

Like with welding three solids together on a rough edge, welding any number of solids on rough boundaries to form the graph $G$, does not change the energy barrier of the $\bar{X}$ operator. 

\begin{lemma}
\label{noChange2FlatRegions}
Doing $Z$-type welds between rough boundaries $V=\{i\}$ of solid codes $E=\{(i,j)\}$ in a graph $G=(V,E)$ does not change the energy barrier of the membrane operator.
\end{lemma}

\begin{proof}
The graph of flat-$Z$ regions does not change because $Z$-type quasi-particles can pass through rough welds without creating new quasi-particles while simultaneously not passing through smooth welds where additional quasi-particles are created. Using the parity lower bound, we get the same energy barrier, which can be saturated if quasi-particles annihilate as soon as they are created. $\Box$
\end{proof}

\begin{corollary}By lemma \ref{noChange2FlatRegions} we can conclude that the welded solid code $G$ has an energy barrier of $O(d)$ for $\bar{X}$, for any solid code that is $O(d)$ qubits wide.
\end{corollary}

Notice that each solid $(i,j)$ is a region where $X$-type quasi-particles can move around without creating new quasi-particles. That is, $(i,j)$ labels flat-$X$ regions and i and j label rough boundaries where new quasi-particles are created. By the parity lower bound, the energy barrier of the welded solid code $G$ is lower bounded by the Ising model corresponding to the same graph $G$ with the reinterpretation that vertices $V$ are qubits and edges $E$ are $ZZ$ interactions. This is precisely how we construct the code with a power law energy barrier that proves theorem \ref{weldedCodeTheorem}, namely that a 3-d code, with side length $L$, with an energy barrier of $O(L^{\frac{2}{3}})$ exists.

\begin{proof}[Proof of theorem \ref{weldedCodeTheorem}:]
We find the energy barrier of $\bar{X}$ and $\bar{Z}$ and then tune them to maximize the energy barrier. We make the graph of solids, $G$, a 3-d cubic lattice that is $R$ solids wide. By the parity lower bound the energy barrier of $\bar{Z}$ will be lower bounded by $O(R^2)$, like for the Ising model in three-dimensions.  This bound can be saturated if in each solid we annihilate pairs of $X$-type quasi-particles as soon as they are created. Simultaneously the energy barrier for the membrane operator in the solid code is $O(d)$, by lemma \ref{solidCodeEnergyBarrier}, and does not change after welding into the graph $G$ by lemma \ref{noChange2FlatRegions}.  If the entire code is $L$ qubits wide, each solid code is $d$ qubits wide and the welded solid code is $R$ solids wide then the number of qubits is $N=O(L^3)=O(d^3 R^3)$. The logical-$Z$ and logical-$X$ operators have energy barriers of $O(R^2)$ and $O(d)$ respectively with distances of $O(R^3 d)$ and $O(d^2)$ respectively. If we relate the two scales $d$ and $R$ by requiring that $d=O(R^\alpha)$ then the energy barriers of the logical-$X$ and logical-$Z$ operators will be $O(R^2)=O(N^{\frac{2}{3(1+\alpha)}})$ and $O(d)=L(N^{\frac{\alpha}{3(\alpha+1)}})$ respectively. The energy barrier is maximized when the two energy barriers are equal since they have opposite slopes with respect to $\alpha$. The maximum happens when $\alpha=2$. Hence the minimum energy barrier for the welded solid code can be tuned to be $O(N^{\frac{2}{9}})=O(L^{\frac{2}{3}})$ with a minimum distance of $\min(d^2,dR^3)=O(L^{\frac{4}{3}})$.  $\Box$
\end{proof}

\section{Discussion}

Welding gives a constructive procedure for producing new codes, i.e. one can simply combine the shape of two logical operators, and it allows one to easily consider more general embeddings of qubits other than periodic lattices. As it was shown in this paper, welding can be used to combine spatially local, bounded density codes, to produce other spatially local, bounded density, codes while simultaneously creating a larger energy barrier by causing bifurcations in the logical operators.  

Although the welded solid code has a power law energy barrier, it does not have a phase transition \cite{kamil2012thermowelding}.  It seems likely that, like the cubic code \cite{haah2012quantum}, the welded solid code will have an increased storage time up to a maximum size where entropic effects from the bifurcating string like operator takes over.

The shape of the logical operators is crucial. It was shown \cite{yoshida2011feasibility} that in local translation-invariant codes with a bounded number of qubits, that logical operators are either points, strings, planes or volumes.  The logical operators $\bar{X}$ and $\bar{Z}$ must overlap at a single position to anticommute.  Strings and points have constant energy barriers and only points and lines can intersect volumes and flat sheets at a single point.  Two flat sheets can intersect at a point, but only in four-dimension or higher as in the 4-d toric code \cite{dennis2002topological}.  

The welded solid code gets around the problem of geometry of logical operators by having the code be defined on both a fine grained scale and a course grained scale.  The fine grained scale gives a membrane operator and bifurcating string like operators. When course graining, the membrane operator becomes a point-like operator and the bifurcating string like operator looks like a volume operator. 

Another way around the problem of geometry is to use curved logical operators.  Two curved sheets can intersect at a single point in three dimensions.  For instance, if the logical-$\bar{X}$ operator was a sheet on the surface of a sphere then we could make $\bar{Z}$ have the shape of a cone that intersects the surface of the sphere at a single qubit. This sphere-cone construction also gives an energy barrier of $O(N^{\frac{2}{9}})$ which is the same as the welded solid code. Likely an energy barrier of $O(L)=O(N^{1/3})$ will be necessary, like in the 2-d Ising model, for a phase transition.   

Yet another way of viewing the welded solid code is as a five dimensional homology code embedded in three dimensions. I say this because if we had used the a 2-d square graph for welding solids together then we would have had something equivalent to the 4-d toric code embedded in 3-d in the sense that if the solids were two qubits tall, then we would have star operators for both $X$ and $Z$-type stabilizers. We might be able to make D-dimensional homology codes local in three dimensions using welding.

Both the sphere-cone construction and the welded solid codes require that there are large numbers of qubits that are in the same space yet do not interact which could be a problem for implementing such a Hamiltonian experimentally.

In this paper, welding has been primarily applied to toric-like codes in different dimensions. This was because surface codes and solid codes give clearly defined flat regions that are easy to analyze. However, welding is a general procedure and it will be interested to see what different types of codes can be constructed via welding?

\paragraph{Acknowledgments} Thank you Aram Harrow and Steve Flammia for the stimulating discussions and helpful suggestions throughout the writing process. Thank you Aram Harrow for helping to clarify the parity bound.  This work was supported by NSF grant 0829937, DARPA QuEST contract FA9550-09-1-0044 and IARPA via DoI NBC contract D11PC20167.  The U.S. Government is authorized to reproduce and distribute reprints for Governmental purposes notwithstanding any copyright annotation thereon. The views and conclusions contained herein are those of the author and should not be interpreted as necessarily representing the official policies or endorsements, either expressed or implied, of IARPA, DoI/NBC, or the U.S. Government.

\bibliography{refs}
\bibliographystyle{plain}

\end{document}